\documentclass[pra,twocolumn,amsmath,18pt]{revtex4-1}
\usepackage{pifont}
\usepackage{graphicx,epsfig,subfigure,dsfont,amssymb,amsmath,amsthm,amsfonts,amsbsy,mathrsfs,amscd}
\usepackage[all]{xy}

\newtheorem{theorem}{Theorem}

\newtheorem{lemma}{Lemma}

\input amssym.def

\begin{document}

\title{Tighter constraints of multiqubit entanglement}

\smallskip
\author{Long-Mei Yang$^1$ }
\author{Bin Chen$^2$}
\thanks{Corresponding author: chenbin5134@163.com}
\author{Shao-Ming Fei$^{1,3}$}
\author{Zhi-Xi Wang$^1$}
\thanks{Corresponding author: wangzhx@cnu.edu.cn}
\affiliation{
$^1$School of Mathematical Sciences, Capital Normal University, Beijing 100048, China\\
$^2$College of Mathematical Science, Tianjin Normal University, Tianjin 300387, China\\
$^3$Max-Planck-Institute for Mathematics in the Sciences, 04103 Leipzig, Germany
}

\begin{abstract}
Monogamy and polygamy relations characterize the distributions of entanglement in multipartite systems.
We provide classes of monogamy and polygamy inequalities of multiqubit entanglement
in terms of concurrence, entanglement of formation, negativity, Tsallis-$q$ entanglement and R\'{e}nyi-$\alpha$ entanglement, respectively.
We show that these inequalities are tighter than the existing ones for some classes of quantum states.
\end{abstract}

\maketitle
\baselineskip20pt
\section{Introduction}

Quantum entanglement is an essential feature of quantum mechanics which distinguishes the quantum from the classical world
and plays a very important role in quantum information processing \cite{dynamics,GaoYang,demonstration,QIp}.
One singular property of quantum entanglement is that a quantum system entangled with one of the other subsystems limits
its entanglement with the remaining ones, known as the monogamy of entanglement (MoE) \cite{BMT,JSK}.
MoE plays a key role in  many quantum information and communication processing tasks such as the security proof in quantum cryptographic scheme \cite{CHB}
and the security analysis of quantum key distribution \cite{Pawl}.

For a tripartite quantum state $\rho_{ABC}$, MoE can be described as the following inequality
\begin{equation}
\mathcal{E}(\rho_{A|BC})\geq \mathcal{E}(\rho_{AB})+\mathcal{E}(\rho_{AC}),
\end{equation}
where $\rho_{AB}={\rm tr}_C(\rho_{ABC})$ and $\rho_{AC}={\rm tr}_B(\rho_{ABC})$ are reduced density matrices, and $\mathcal{E}$ is an entanglement measure.
However, it has been shown that not all entanglement measures satisfy such monogamy relations.
It has been shown that the squared concurrence $\mathcal{C}^2$ \cite{TJO,Bai}, the squared entanglement of formation (EoF) $E^2$ \cite{Oliveira} and the squared convex-roof extended negativity (CREN) $\mathcal{N}_c^2$ \cite{JSK8,Feng1} satisfy the monogamy relations for multiqubit states.

Another important concept is the assisted entanglement, which is a dual amount to bipartite entanglement measure.
It has a dually monogamous property in multipartite quantum systems and gives rise to polygamy relations.
For a tripartite state $\rho_{ABC}$, the usual polygamy relation is of the form,
\begin{equation}
\mathcal{E}^a(\rho_{A|BC})\leq\mathcal{E}^a (\rho_{AB})+\mathcal{E}^a(\rho_{AC}),
\end{equation}
where $\mathcal{E}^a$ is the corresponding entanglement measure of assistance associated to $\mathcal{E}$.
Such polygamy inequality has been deeply investigated in recent years,
and was generalized to multiqubit systems and classes of higher-dimensional quantum systems \cite{JSK3,F.B,G.G1,G.G2,JSK5,JSK6,JSK8,Song}.

Recently, generalized classes of monogamy inequalities related to the $\beta$th power of entanglement measures were proposed.
In Ref. \cite{SM.Fei1,SM.Fei3}, the authors proved that the squared concurrence and CREN satisfy the monogamy inequalities in multiqubit systems for $\beta\geq2$.
It has also been shown that the EoF satisfies monogamy relations when $\beta\geq\sqrt{2}$ \cite{SM.Fei1,SM.Fei2,SM.Fei3}.
Besides, the Tsallis-$q$ entanglement and R\'enyi-$\alpha$ entanglement satisfy monogamy relations when $\beta\geq1$ \cite{JSK2,JSK3,SM.Fei2,SM.Fei3}
for some cases.
Moreover, the corresponding polygamy relations have also been established \cite{G.G1,G.G2,JSK5,JSK9,Song,Yongming}.

In this paper, we investigate monogamy relations and polygamy relations in multiqubit systems.
We provide tighter constraints of multiqubit entanglement than all the existing ones,
thus give rise to finer characterizations of the entanglement distributions among the multiqubit systems.

\section{Tighter constraints related to concurrence}

We first consider the monogamy inequalities and polygamy inequalities for concurrence.
For a bipartite pure state $|\psi\rangle_{AB}$ in Hilbert space $H_A\otimes H_B$, the concurrence is defined as \cite{fide,SM.Fei4}
$\mathcal{C}(|\psi\rangle_{AB})=\sqrt{2(1-{\rm tr}\rho_A^2)}$ with $\rho_A={\rm tr}_B|\psi\rangle_{AB}\langle\psi|$.
The concurrence for a bipartite mixed state $\rho_{AB}$ is defined by the convex roof extension,
$\mathcal{C}(\rho_{AB})=\min\limits_{\{p_i,|\psi_i\rangle\}}\sum\limits_{i}p_i\mathcal{C}(|\psi_i\rangle)$,
where the minimum is taken over all possible decompositions of $\rho_{AB}=\sum\limits_{i}p_i|\psi_i\rangle\langle\psi_i|$
with $\sum p_i=1$ and $p_i\geq0$.
For an $N$-qubit state $\rho_{AB_1\cdots B_{N-1}}\in H_A\otimes H_{B_1}\otimes\cdots\otimes H_{B_{N-1}}$,
the concurrence $\mathcal{C}(\rho_{A|B_1\cdots B_{N-1}})$ of the state $\rho_{AB_1\cdots B_{N-1}}$ under bipartite partition $A$ and $B_1\cdots B_{N-1}$ satisfies \cite{SM.Fei1}
\begin{equation}\label{Con1}
\begin{array}{rl}
&\mathcal{C}^\beta(\rho_{A|B_1\cdots B_{N-1}})\\
&\ \ \geq\mathcal{C}^\beta(\rho_{AB_1})+\mathcal{C}^\beta(\rho_{AB_2})+\cdots+\mathcal{C}^\beta(\rho_{AB_{N-1}}),
\end{array}
\end{equation}
for $\beta\geq2$, where $\rho_{AB_j}$ denote two-qubit reduced density matrices of subsystems $AB_j$ for $j=1,2,\ldots,N-1$.
Later, the relation \eqref{Con1} is improved for the case $\beta\geq2$ \cite{SM.Fei2} as
\begin{equation}\label{Con2}
\begin{array}{rl}
&\mathcal{C}^\beta(\rho_{A|B_1\cdots B_{N-1}})\\
&\ \ \geq\mathcal{C}^\beta(\rho_{AB_1})+\frac{\beta}{2}\mathcal{C}^\beta(\rho_{AB_2})+\cdots\\
&\ \ \ \ +\big(\frac{\beta}{2}\big)^{m-1}\mathcal{C}^\beta(\rho_{AB_{m}})\\
&\ \ \ \ +\big(\frac{\beta}{2}\big)^{m+1}[\mathcal{C}^\beta(\rho_{AB_{m+1}})+\cdots+\mathcal{C}^\beta(\rho_{AB_{N-2}})]\\
&\ \ \ \ +\big(\frac{\beta}{2}\big)^m\mathcal{C}^\beta(\rho_{AB_{N-1}})
\end{array}
\end{equation}
conditioned that $\mathcal{C}(\rho_{AB_i})\geq\mathcal{C}(\rho_{A|B_{i+1}\cdots B_{N-1}})$ for $i=1,2,\ldots,m$,
and $\mathcal{C}(\rho_{AB_j})\leq\mathcal{C}(\rho_{A|B_{j+1}\cdots B_{N-1}})$ for $j=m+1,\ldots,N-2$.
The relation \eqref{Con2} is further improved for $\beta\geq2$ as \cite{SM.Fei3}
\begin{equation}\label{Con3}
\begin{array}{rl}
&\mathcal{C}^\beta(\rho_{A|B_1\cdots B_{N-1}})\\
&\ \ \geq\mathcal{C}^\beta(\rho_{AB_1})+\big(2^{\frac{\beta}{2}}-1\big)\mathcal{C}^\beta(\rho_{AB_2})+\cdots\\
&\ \ \ \ +\big(2^{\frac{\beta}{2}}-1\big)^{m-1}\mathcal{C}^\beta(\rho_{AB_{m}})\\
&\ \ \ \ +\big(2^{\frac{\beta}{2}}-1\big)^{m+1}[\mathcal{C}^\beta(\rho_{AB_{m+1}})+\cdots+\mathcal{C}^\beta(\rho_{AB_{N-2}})]\\
&\ \ \ \ +\big(2^{\frac{\beta}{2}}-1\big)^m\mathcal{C}^\beta(\rho_{AB_{N-1}})
\end{array}
\end{equation}
with the same conditions as in \eqref{Con2}.

For a tripartite state $|\psi\rangle_{ABC}$, the concurrence of assistance (CoA) is defined by \cite{EOS,monogamy}
\begin{equation}
\mathcal{C}_a(\rho_{AB})=\max\limits_{\{p_i,|\psi_i\rangle\}}\sum\limits_{i}p_i\mathcal{C}(|\psi_i\rangle),
\end{equation}
where the maximun is taken over all possible pure state decompositions of $\rho_{AB}$, and $\mathcal{C}(|\psi\rangle_{AB})=\mathcal{C}_a(|\psi\rangle_{AB})$.
The generalized polygamy relation based on the concurrence of assistance was established in \cite{G.G1,G.G2}
\begin{equation}\label{Con14}
\begin{array}{rl}
&\mathcal{C}^2(|\psi\rangle_{A|B_1\cdots B_{N-1}})\\
&\ \ =\mathcal{C}_a^2(|\psi\rangle_{A|B_1\cdots B_{N-1}})\\
&\ \ \leq\mathcal{C}_a^2(\rho_{AB_1})+\mathcal{C}_a^2(\rho_{AB_2})+\cdots+\mathcal{C}_a^2(\rho_{AB_{N-1}}).
\end{array}
\end{equation}

These monogamy and polygamy relations for concurrence can be further tightened under some conditions.
To this end, we first introduce the following lemma.
\begin{lemma}\label{con1}
Suppose that $k$ is a real number satisfying $0< k\leq1$, then for any $0\leq t\leq k$ and non-negative real numbers $m,n$, we have
\begin{equation}\label{Con4}
(1+t)^m\geq1+\frac{(1+k)^m-1}{k^m}t^m
\end{equation}
for $m\geq1$, and
\begin{equation}\label{Con5}
(1+t)^n\leq1+\frac{(1+k)^n-1}{k^n}t^n
\end{equation}
for $0\leq n\leq 1$.
\end{lemma}
\begin{proof}
We first consider the function $f(m,x)=(1+x)^m-x^m$ with $x\geq\frac{1}{k}$ and $m\geq1$.
Then $f(m,x)$ is an increasing function of $x$, since $\frac{\partial f(m,x)}{\partial x}=m[(1+x)^{m-1}-x^{m-1}]\geq 0$.
Thus,
\begin{equation}\label{Con8}
f(m,x)\geq f(m,\frac{1}{k})=\big(1+\frac{1}{k}\big)^m-\big(\frac{1}{k}\big)^m=\frac{(k+1)^m-1}{k^m}.
\end{equation}
Set $x=\frac{1}{t}$ in \eqref{Con8}, we get the inequality \eqref{Con4}.

Similar to the proof of inequality \eqref{Con4}, we can obtain the inequality \eqref{Con5},
since in this case $f(n,x)$ is a decreasing function of $x$ for $x\geq \frac{1}{k}$ and $0\leq n\leq1$.
\end{proof}

In the next, we denote $\mathcal{C}_{AB_i}=\mathcal{C}(\rho_{AB_i})$ the concurrence of $\rho_{AB_i}$
and $\mathcal{C}_{A|B_1\cdots B_{N-1}}=\mathcal{C}(\rho_{A|B_1\cdots B_{N-1}})$ for convenience.

\begin{lemma}\label{con2}
Suppose that $k$ is a real number satisfying $0< k\leq1$.
Then for any $2\otimes2\otimes2^{n-2}$ mixed state $\rho\in H_A\otimes H_B\otimes H_C$,
if $\mathcal{C}_{AC}^2\leq k\mathcal{C}_{AB}^2$, we have
\begin{equation}\label{Con9}
\mathcal{C}_{A|BC}^\beta\geq\mathcal{C}_{AB}^\beta+\frac{(1+k)^{\frac{\beta}{2}}-1}{k^{\frac{\beta}{2}}}\mathcal{C}_{AC}^\beta,
\end{equation}
for all $\beta\geq 2$.
\end{lemma}
\begin{proof}
Since $\mathcal{C}_{AC}^2\leq k\mathcal{C}_{AB}^2$ and $\mathcal{C}_{AB}>0$,
we obtain
\begin{equation}
\begin{array}{rl}
&\mathcal{C}_{A|BC}^\beta\geq(\mathcal{C}_{AB}^2+\mathcal{C}_{AC}^2)^{\frac{\beta}{2}}\\[1.5mm]
&\ \ \ \ \ \ \ \ \ =\mathcal{C}_{AB}^\beta \Big(1+\frac{\mathcal{C}_{AC}^2}{\mathcal{C}_{AB}^2}\Big)^{\frac{\beta}{2}}\\[1.5mm]
&\ \ \ \ \ \ \ \ \ \geq\mathcal{C}_{AB}^\beta \Big[1+\frac{(1+k)^{\frac{\beta}{2}}-1}{k^{\frac{\beta}{2}}}
\Big(\frac{\mathcal{C}_{AC}^2}{\mathcal{C}_{AB}^2}\Big)^{\frac{\beta}{2}}\Big]\\[3mm]
&\ \ \ \ \ \ \ \ \ =\mathcal{C}_{AB}^\beta+\frac{(1+k)^{\frac{\beta}{2}}-1}{k^{\frac{\beta}{2}}}\mathcal{C}_{AC}^\beta,
\end{array}
\end{equation}
where the first inequality is due to the fact, $\mathcal{C}_{A|BC}^2\geq\mathcal{C}_{AB}^2+\mathcal{C}_{AC}^2$ for
arbitrary $2\otimes2\otimes2^{n-2}$ tripartite state $\rho_{ABC}$ \cite{TJO,XJR} and
the second is due to Lemma \ref{con1}.
We can also see that if $\mathcal{C}_{AB}=0$, then $\mathcal{C}_{AC}=0$, and the lower bound becomes trivially zero.
\end{proof}

For multiqubit systems, we have the following Theorems.
\begin{theorem}\label{concurrence1}
Suppose $k$ is a real number satisfying $0<k\leq 1$.
For an $N$-qubit mixed state $\rho_{AB_1\cdots B_{N-1}}$, if $k\mathcal{C}_{AB_i}^2\geq\mathcal{C}_{A|B_{i+1}\cdots B_{N-1}}^2$
for $i=1,2,\ldots,m$, and $\mathcal{C}_{AB_j}^2\leq k\mathcal{C}_{A|B_{j+1}\cdots B_{N-1}}^2$ for $j=m+1,\ldots,N-2$,
$\forall 1\leq m\leq N-3$, $N\geq 4$, then we have
\begin{equation}\label{Con12}
\begin{array}{rl}
&\mathcal{C}^\beta_{A|B_1\cdots B_{N-1}}\\[2.0mm]
&\ \ \geq\mathcal{C}^\beta_{AB_1}+\frac{(1+k)^{\frac{\beta}{2}}-1}{k^{\frac{\beta}{2}}}\mathcal{C}^\beta_{AB_2}+\cdots\\[2.0mm]
&\ \ \ \ +\Big(\frac{(1+k)^{\frac{\beta}{2}}-1}{k^{\frac{\beta}{2}}}\Big)^{m-1}\mathcal{C}^\beta_{AB_{m}}\\[2.0mm]
&\ \ \ \ +\Big(\frac{(1+k)^{\frac{\beta}{2}}-1}{k^{\frac{\beta}{2}}}\Big)^{m+1}\Big(\mathcal{C}^\beta_{AB_{m+1}}+\cdots+
\mathcal{C}^\beta_{AB_{N-2}}\Big)\\[2.0mm]
&\ \ \ \ +\Big(\frac{(1+k)^{\frac{\beta}{2}}-1}{k^{\frac{\beta}{2}}}\Big)^m\mathcal{C}^\beta_{AB_{N-1}}
\end{array}
\end{equation}
for all $\beta\geq2$.
\end{theorem}
\begin{proof}
From the inequality \eqref{Con9}, we have
\begin{equation}\label{Con10}
\begin{array}{rl}
&\mathcal{C}_{A|B_1B_2\cdots B_{N-1}}\\[2.0mm]
&\ \ \geq\mathcal{C}_{AB_1}^\beta+\frac{(1+k)^{\frac{\beta}{2}}-1}{k^{\frac{\beta}{2}}}\mathcal{C}_{A|B_2\cdots B_{N-1}}^\beta\\[2.0mm]
&\ \ \geq\mathcal{C}_{AB_1}^\beta+\frac{(1+k)^{\frac{\beta}{2}}-1}{k^{\frac{\beta}{2}}}\mathcal{C}_{AB_2}^\beta\\[2.0mm]
&\ \ \ \ +\Big(\frac{(1+k)^{\frac{\beta}{2}}-1}{k^{\frac{\beta}{2}}}\Big)^2\mathcal{C}_{A|B_3\cdots B_{N-1}}^\beta\\[2.0mm]
&\ \ \geq\cdots\\[2.0mm]
&\ \ \geq\mathcal{C}_{AB_1}^\beta+\Big(\frac{(1+k)^{\frac{\beta}{2}}-1}{k^{\frac{\beta}{2}}}\Big)\mathcal{C}_{AB_2}^\beta\\[2.0mm]
&\ \ \ \ +\cdots +\Big(\frac{(1+k)^{\frac{\beta}{2}}-1}{k^{\frac{\beta}{2}}}\Big)^{m-1}\mathcal{C}_{AB_m}^\beta\\[2.0mm]
&\ \ \ \ +\Big(\frac{(1+k)^{\frac{\beta}{2}}-1}{k^{\frac{\beta}{2}}}\Big)^{m}\mathcal{C}_{A|B_{m+1}\cdots B_{N-1}}^\beta.
\end{array}
\end{equation}
Since $\mathcal{C}_{AB_j}^2\leq k\mathcal{C}_{A|B_{j+1}\cdots B_{N-1}}^2$,
for $j=m+1,\ldots,N-2$, we get
\begin{equation}\label{Con11}
\begin{array}{rl}
&\mathcal{C}_{A|B_{m+1}\cdots B_{N-1}}^\beta\\[2.0mm]
&\ \ \geq\frac{(1+k)^{\frac{\beta}{2}}-1}{k^{\frac{\beta}{2}}}\mathcal{C}_{AB_{m+1}}^\beta+\mathcal{C}_{A|B_{m+2\cdots B_{N-1}}}^\beta\\[2.0mm]
&\ \ \geq\frac{(1+k)^{\frac{\beta}{2}}-1}{k^{\frac{\beta}{2}}}\Big(\mathcal{C}_{AB_{m+1}}^\beta+\cdots+\mathcal{C}_{AB_{N-2}}^\beta\Big)
+\mathcal{C}_{AB_{N-1}}^\beta.
\end{array}
\end{equation}
Combining \eqref{Con10} and \eqref{Con11}, we get the inequality \eqref{Con12}.
\end{proof}

If we replace the conditions $k\mathcal{C}_{AB_i}\geq\mathcal{C}_{A|B_{i+1}\cdots B_{N-1}}$ for $i=1,2,\ldots,m$,
and $\mathcal{C}_{AB_j}^2\leq k\mathcal{C}_{A|B_{j+1}\cdots B_{N-1}}^2$ for $j=m+1,\ldots,N-2$, $\forall 1\leq m\leq N-3$, $N\geq 4$,
in Theorem \ref{concurrence1} by $k\mathcal{C}_{AB_i}^2\geq \mathcal{C}_{A|B_{i+1}\cdots B_{N-1}}^2$ for $i=1,2,\ldots,N-2$,
then we have the following theorem.
\begin{theorem}\label{concurrence2}
Suppose $k$ is a real number satisfying $0<k\leq 1$.
For an $N$-qubit mixed state $\rho_{AB_1\cdots B_{N-1}}$,
if $k\mathcal{C}_{AB_i}^2\geq\mathcal{C}_{A|B_{i+1}\cdots B_{N-1}}^2$ for all $i=1,2,\ldots,N-2$, then we have
\begin{equation}\label{Con13}
\begin{array}{rcl}
\mathcal{C}^\beta_{A|B_1\cdots B_{N-1}}
&\geq&\mathcal{C}^\beta_{AB_1}+\frac{(1+k)^{\frac{\beta}{2}}-1}{k^{\frac{\beta}{2}}}
\mathcal{C}^\beta_{AB_2}+\cdots\\[3mm]
&&+\Big(\frac{(1+k)^{\frac{\beta}{2}}-1}{k^{\frac{\beta}{2}}}\Big)^{N-2}\mathcal{C}^\beta_{AB_{N-1}}
\end{array}
\end{equation}
for $\beta\geq2$.
\end{theorem}

It can be seen that the inequalities \eqref{Con12} and \eqref{Con13} are tighter than the ones given in Ref. \cite{SM.Fei3},
since
$$
\frac{(1+k)^\frac{\beta}{2}-1}{k^\frac{\beta}{2}}\geq2^{\frac{\beta}{2}}-1
$$
for $\beta\geq2$ and $0<k\leq1$.
The equality holds when $k=1$. Namely, the result (\ref{Con3}) given in \cite{SM.Fei3} are just special cases of ours for $k=1$.
As $\frac{(1+k)^\frac{\beta}{2}-1}{k^{\frac{\beta}{2}}}$ is a decreasing function with respect to $k$ for $0< k\leq 1$ and $\beta\geq2$,
we find that the smaller $k$ is, the tighter the inequalities \eqref{Con9}, \eqref{Con12} and \eqref{Con13} are.

\smallskip
\noindent{\bf Example 1} \, \ Consider the three-qubit state $|\psi\rangle_{ABC}$ in generalized Schmidt decomposition form \cite{Schmidt,SM.Fei5},
\begin{equation}\label{Con6}
|\psi\rangle_{ABC}=\lambda_0|000\rangle+\lambda_1e^{i\varphi}|100\rangle+\lambda_2|101\rangle+\lambda_3|110\rangle+\lambda_4|111\rangle,
\end{equation}
where $\lambda_i\geq0$, $i=,1,2...,4$, and $\sum\limits_{i=0}^4\lambda_i^2=1$.
Then we get $\mathcal{C}_{A|BC}=2\lambda_0\sqrt{\lambda_2^2+\lambda_3^2+\lambda_4^2}$, $\mathcal{C}_{AB}=2\lambda_0\lambda_2$ and
$\mathcal{C}_{AC}=2\lambda_0\lambda_3$.
Set $\lambda_0=\lambda_3=\frac{1}{2},\ \lambda_2=\frac{\sqrt{2}}{2}$ and $\lambda_1=\lambda_4=0$.
We have $\mathcal{C}_{A|BC}=\frac{\sqrt{3}}{2}$, $\mathcal{C}_{AB}=\frac{\sqrt{2}}{2}$ and $\mathcal{C}_{AC}=\frac{1}{2}$.
Then $\mathcal{C}_{AB}^\beta+\big(2^{\frac{\beta}{2}}-1\big)\mathcal{C}_{AC}^\beta=\big(\frac{\sqrt{2}}{2}\big)^\beta+
\big(2^{\frac{\beta}{2}}-1\big)\big(\frac{1}{2}\big)^\beta$ and
$\mathcal{C}_{AB}^\beta+\frac{(1+k)^\frac{\beta}{2}-1}{k^{\frac{\beta}{2}}}\mathcal{C}_{AC}^\beta=\big(\frac{\sqrt{2}}{2}\big)^\beta+
\frac{(1+k)^\frac{\beta}{2}-1}{k^{\frac{\beta}{2}}}\big(\frac{1}{2}\big)^\beta$.
One can see that our result is better than the result (\ref{Con3}) in \cite{SM.Fei3} for $\beta\geq2$, hence better than (\ref{Con1}) and (\ref{Con2}) given in \cite{SM.Fei1,SM.Fei2},
see Fig. 1.
\begin{figure}
  \centering
  \includegraphics[width=6cm]{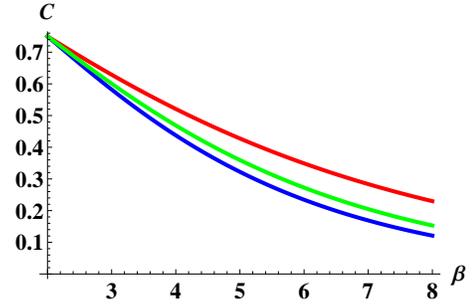}
  \caption{The $y$ axis is the lower bound of the concurrence $\mathcal{C}^\beta_{A|BC}$.
  The red (green) line represents the lower bound from our result for $k=0.6$ ($k=0.8$),
  and the blue line represents the lower bound of (\ref{Con3}) from \cite{SM.Fei3}.}
\end{figure}

We now discuss the polygamy relations for the CoA of $\mathcal{C}_a(|\psi\rangle_{A|B_1\cdots B_{N-1}})$ for $0\leq\beta\leq 2$.
We have the following Theorem.
\begin{theorem}
Suppose $k$ is a real number satisfying $0<k\leq 1$.
For an $N$-qubit pure state $|\psi\rangle_{AB_1\cdots B_{N-1}}$,
if $k\mathcal{C}^2_{aA|B_i}\geq\mathcal{C}^2_{aA|B_{i+1}\cdots B_{N-1}}$
for $i=1,2,\ldots,m$, and $\mathcal{C}^2_{aA|B_j}\leq k\mathcal{C}^2_{aA|B_{j+1}\cdots B_{N-1}}$ for $j=m+1,\ldots,N-2$,
$\forall 1\leq m\leq N-3$, $N\geq 4$, then we have
\begin{equation}\label{Con16}
\begin{array}{rl}
&\mathcal{C}_a^\beta(|\psi\rangle_{A|B_1\cdots B_{N-1}})\\[2.0mm]
&\ \ \leq\mathcal{C}^\beta_{aAB_1}+\frac{(1+k)^{\frac{\beta}{2}}-1}{k^{\frac{\beta}{2}}}\mathcal{C}^\beta_{aAB_2}+\cdots\\[2.0mm]
&\ \ \ \ +\Big(\frac{(1+k)^{\frac{\beta}{2}}-1}{k^{\frac{\beta}{2}}}\Big)^{m-1}\mathcal{C}^\beta_{aAB_{m}}\\[2.0mm]
&\ \ \ \ +\Big(\frac{(1+k)^{\frac{\beta}{2}}-1}{k^{\frac{\beta}{2}}}\Big)^{m+1}\Big(\mathcal{C}^\beta_{aAB_{m+1}}+
\cdots+\mathcal{C}^\beta_{aAB_{N-2}}\Big)\\[2.0mm]
&\ \ \ \ +\Big(\frac{(1+k)^{\frac{\beta}{2}}-1}{k^{\frac{\beta}{2}}}\Big)^m\mathcal{C}^\beta_{aAB_{N-1}}
\end{array}
\end{equation}
for all $0\leq\beta\leq2$.
\end{theorem}
\begin{proof}
The proof is similar to the proof of Theorem \ref{concurrence1} by using inequality \eqref{Con5}.
\end{proof}

\begin{theorem}
Suppose $k$ is a real number satisfying $0<k\leq 1$.
For an $N$-qubit pure state $|\psi\rangle_{AB_1\cdots B_{N-1}}$,
if $k\mathcal{C}^2_{aAB_i}\geq\mathcal{C}^2_{aA|B_{i+1}\cdots B_{N-1}}$ for all $i=1,2,\ldots,N-2$, then we have
\begin{equation}\label{Con17}
\begin{array}{rcl}
\mathcal{C}^\beta_a{|\psi\rangle_{A|B_1\cdots B_{N-1}}}
&\leq&\mathcal{C}^\beta_{aAB_1}+\frac{(1+k)^{\frac{\beta}{2}}-1}
{k^{\frac{\beta}{2}}}\mathcal{C}^\beta_{aAB_2}+\cdots\\[3.0mm]
&&+\Big(\frac{(1+k)^{\frac{\beta}{2}}-1}{k^{\frac{\beta}{2}}}\Big)^{N-2}\mathcal{C}^\beta_{aAB_{N-1}}
\end{array}
\end{equation}
for $0\leq\beta\leq2$.
\end{theorem}
The inequalities \eqref{Con16} and \eqref{Con17} are also upper bounds of $\mathcal{C}(|\psi\rangle_{A|B_1\cdots B_{N-1}})$
for pure state $|\psi\rangle_{AB_1\cdots B_{N-1}}$
since $\mathcal{C}(|\psi\rangle_{A|B_1\cdots B_{N-1}})=\mathcal{C}_a(|\psi\rangle_{A|B_1\cdots B_{N-1}})$.

\section{Tighter constraints relate to EoF}
Let $H_A$ and $H_B$ be two Hilbert spaces with dimension $m$ and $n$ $(m\leq n)$, respectively.
Then the entanglement of formation (EoF) \cite{CHB1,CHB2} is defined as follows:
for a pure state $|\psi\rangle_{AB}\in H_A\otimes H_B$, the EoF is given by
\begin{equation}
E(|\psi\rangle_{AB})=\mathcal{S}(\rho_A),
\end{equation}
where $\rho_A={\rm Tr}_B(|\psi\rangle_{AB}\langle\psi|)$ and $\mathcal{S}(\rho)=-{\rm Tr}(\rho\log_2\rho)$.
For a bipartite mixed state $\rho_{AB}\in H_A\otimes H_B$, the EoF is given by
\begin{equation}
E(\rho_{AB})=\min\limits_{\{p_i,|\psi_i\rangle\}}\sum\limits_{i}p_iE(|\psi_i\rangle),
\end{equation}
with the minimum taking over all possible pure state decomposition of $\rho_{AB}$.

In Ref. \cite{EoF}, Wootters showed that $E(|\psi\rangle)=f(\mathcal{C}^2(|\psi\rangle))$ for $2\otimes m \ (m\geq2)$ pure state $|\psi\rangle$, and $E(\rho)=f(\mathcal{C}^2(\rho))$ for two-qubit mixed state $\rho$,
where $f(x)=H\big(\frac{1+\sqrt{1-x}}{2}\big)$ and $H(x)=-x\log_2x-(1-x)\log_2(1-x)$.
$f(x)$ is a monotonically increasing function for $0\leq x\leq1$, and satisfies the following relations:
\begin{equation}\label{EoF4}
f^{\sqrt{2}}(x^2+y^2)\geq f^{\sqrt{2}}(x^2)+f^{\sqrt{2}}(y^2),
\end{equation}
where $f^{\sqrt{2}}(x^2+y^2)=[f(x^2+y^2)]^{\sqrt{2}}$.

Although EoF does not satisfy the inequality $E_{AB}+E_{AC}\leq E_{A|BC}$ \cite{CKW},
the authors in \cite{BZYW} showed that EoF is a monotonic function satisfying
$E^2(\rho_{A|B_1B_2\cdots B_{N-1}})\geq \sum_{i=1}^{N-1}E^2(\rho_{AB_i})$.
For $N$-qubit systems, one has \cite{SM.Fei1}
\begin{equation}\label{EoF1}
E^\beta_{A|B_1B_2\cdots B_{N-1}}\geq E^\beta_{AB_1}+E^\beta_{AB_2}+\cdots+E^\beta_{AB_{N-1}},
\end{equation}
for $\beta\geq\sqrt{2}$, where $E_{A|B_1B_2\cdots B_{N-1}}$ is the EoF of $\rho$ under bipartite partition $A|B_1B_2\cdots B_{N-1}$,
and $E_{AB_i}$ is the EoF of the mixed state $\rho_{AB_i}={\rm Tr}_{B_1\cdots B_{i-1},B_{i+1}\cdots B_{N-1}}(\rho)$
for $i=1,2,\ldots,N-1$.
Recently, the authors in Ref. \cite{SM.Fei2} proposed a monogamy relation that is tighter than the inequality \eqref{EoF1},
\begin{equation}\label{EoF2}
\begin{array}{rl}
&E^\beta_{A|B_1B_2\cdots B_{N-1}}\\[1.5mm]
&\ \ \geq E^\beta_{AB_1}+\frac{\beta}{\sqrt{2}}E^\beta_{AB_2}+\cdots+\big(\frac{\beta}{\sqrt{2}}\big)^{m-1}E^\beta_{AB_m}\\[1.5mm]
&\ \ \ \ +\big(\frac{\beta}{\sqrt{2}}\big)^{m+1}\big(E^\beta_{AB_{m+1}}+\cdots +E^\beta_{AB_{N-2}}\big)\\[1.5mm]
&\ \ \ \ +\big(\frac{\beta}{\sqrt{2}}\big)^mE^\beta_{AB_{N-1}},
\end{array}
\end{equation}
if $\mathcal{C}_{AB_i}\geq\mathcal{C}_{A|B_{j+1}\cdots B_{N-1}}$ for $i=1,2,\ldots,m$, and $\mathcal{C}_{AB_j}\leq\mathcal{C}_{A|B_{j+1}\cdots B_{N-1}}$
for $j=m+1,\ldots,N-2$, $\forall 1\leq m\leq N-3$, $N\geq4$ for $\beta\geq\sqrt{2}$.
The inequality \eqref{EoF2} is also improved to
\begin{equation}\label{EoF9}
\begin{array}{rl}
&E^\beta_{A|B_1B_2\cdots B_{N-1}}\\[1.5mm]
&\ \ \geq E^\beta_{AB_1}+\Big(2^{\frac{\beta}{\sqrt{2}}}-1\Big)E^\beta_{AB_2}+\cdots+
\Big(2^{\frac{\beta}{\sqrt{2}}}-1\Big)^{m-1}\\[1.5mm]
&\ \ \ \ \times E^\beta_{AB_m}+\Big(2^{\frac{\beta}{\sqrt{2}}}-1\Big)^{m+1}\big(E^\beta_{AB_{m+1}}+\cdots+E^\beta_{AB_{N-2}}\big) \\[1.5mm]
&\ \ \ \ +\Big(2^{\frac{\beta}{\sqrt{2}}}-1\Big)^mE^\beta_{AB_{N-1}},
\end{array}
\end{equation}
under the same conditions as that of inequality \eqref{EoF2}.

In fact, these inequalities can be further improved to even tighter monogamy relations.
\begin{theorem}\label{eof1}
Suppose $k$ is a real number satisfying $0<k\leq 1$.
For any $N$-qubit mixed state $\rho_{AB_1\cdots B_{N-1}}$,
if $kE^{\sqrt{2}}_{AB_i}\geq E^{\sqrt{2}}_{A|B_{i+1}\cdots B_{N-1}}$ for $i=1,2,\ldots,m$, and
$E^{\sqrt{2}}_{AB_j}\leq kE^{\sqrt{2}}_{A|B_{j+1}\cdots B_{N-1}}$ for $j=m+1,\ldots,N-2$, $\forall 1\leq m\leq N-3$, $N\geq4$,
the entanglement of formation $E(\rho)$ satisfies
\begin{equation}\label{EoF3}
\begin{array}{rl}
&E^\beta_{A|B_1B_2\cdots B_{N-1}}\\[2.0mm]
&\ \ \geq E^\beta_{AB_1}+\frac{(1+k)^t-1}{k^t}E^\beta_{AB_2}+\cdots+\Big(\frac{(1+k)^t-1}{k^t}\Big)^{m-1}E^\beta_{AB_m}\\[2.0mm]
&\ \ \ \ +\Big(\frac{(1+k)^t-1}{k^t}\Big)^{m+1}(E^\beta_{AB_{m+1}}+\cdots+E^\beta_{AB_{N-2}})\\[2.0mm]
&\ \ \ \ +\Big(\frac{(1+k)^t-1}{k^t}\Big)^mE^\beta_{AB_{N-1}},
\end{array}
\end{equation}
for $\beta\geq\sqrt{2}$, where $t=\frac{\beta}{\sqrt{2}}$.
\end{theorem}
\begin{proof}
For $\beta\geq\sqrt{2}$ and $k f^{\sqrt{2}}(x^2)\geq f^{\sqrt{2}}(y^2)$, we find
\begin{equation}\label{EoF5}
\begin{array}{rl}
&f^{\beta}(x^2+y^2)=[f^{\sqrt{2}}(x^2+y^2)]^t\\[1.5mm]
& \ \ \ \ \ \ \ \ \ \ \ \ \ \ \ \ \ \geq[f^{\sqrt{2}}(x^2)+f^{\sqrt{2}}(y^2)]^t\\[1.5mm]
& \ \ \ \ \ \ \ \ \ \ \ \ \ \ \ \ \ \geq[f^{\sqrt{2}}(x^2)]^t+\frac{(1+k)^t-1}{k^t}[f^{\sqrt{2}}(y^2)]^t\\[1.5mm]
& \ \ \ \ \ \ \ \ \ \ \ \ \ \ \ \ \ =f^\beta(x^2)+\frac{(1+k)^t-1}{k^t}f^\beta(y^2),
\end{array}
\end{equation}
where the first inequality is due to the inequality \eqref{EoF4}, and the second inequality can be obtained from inequality \eqref{Con4}.

Let $\rho=\sum_ip_i|\psi_i\rangle\langle\psi_i|\in H_A\otimes H_{B_1}\otimes\cdots H_{B_{N-1}}$ be the optimal decomposition of
$E_{A|B_1B_2\cdots B_{N-1}}(\rho)$ for the $N$-qubit mixed state $\rho$. Then \cite{SM.Fei3}
\begin{equation}\label{EoF6}
E_{A|B_1B_2\cdots B_{N-1}}\geq f(\mathcal{C}^2_{A|B_1B_2\cdots B_{N-1}}).
\end{equation}
Thus,
\begin{equation}\label{EoF7}
\begin{array}{rl}
&E^\beta_{A|B_1B_2\cdots B_{N-1}}\\[2.0mm]
&\ \ \geq f^\beta(\mathcal{C}^2_{A|B_1B_2\cdots B_{N-1}})\\[2.0mm]
&\ \ \geq f^\beta(\mathcal{C}^2_{A|B_1})+\frac{(1+k)^t-1}{k^t}f^\beta(\mathcal{C}^2_{A|B_2})+\cdots\\[2.0mm]
&\ \ \ \ +\Big(\frac{(1+k)^t-1}{k^t}\Big)^{m-1}f^\beta(\mathcal{C}^2_{A|B_{m}})+\Big(\frac{(1+k)^t-1}{k^t}\Big)^{m+1}\\[2.0mm]
&\ \ \ \ \ [f^\beta(\mathcal{C}^2_{AB_{m+1}})+\cdots+f^\beta(\mathcal{C}^2_{AB_{N-2}})]\\[2.0mm]
&\ \ \ \ +\Big(\frac{(1+k)^t-1}{k^t}\Big)^m f^\beta(\mathcal{C}^2_{AB_{N-1}})\\[2.0mm]
&\ \ =E^\beta_{AB_1}+\Big(\frac{(1+k)^t-1}{k^t}\Big)E^\beta_{AB_2}+\cdots+\Big(\frac{(1+k)^t-1}{k^t}\Big)^{m-1}\\[2.0mm]
&\ \ \ \ \times E^{m-1}_{AB_m}+\Big(\frac{(1+k)^t-1}{k^t}\Big)^{m+1}(E^\beta_{AB_{m+1}}+\cdots+E^\beta_{AB_{N-2}})\\[2.0mm]
&\ \ \ \ +\Big(\frac{(1+k)^t-1}{k^t}\Big)^m E^\beta_{AB_{N-1}},
\end{array}
\end{equation}
where the first inequality holds due to \eqref{EoF6}, the second inequality is similar to the proof of Theorem \ref{concurrence1} by using inequality \eqref{EoF5},
and the last equality holds since for any $2\otimes 2$ quantum state $\rho_{AB_i}$, $E(\rho_{AB_i})=f[\mathcal{C}^2(\rho_{AB_i})]$.
\end{proof}

Similar to the case of concurrence, we have also the following tighter monogamy relation for EoF.

\begin{theorem}\label{eof2}
Suppose $k$ is a real number satisfying $0<k\leq 1$.
For an $N$-qubit mixed state $\rho_{AB_1\cdots B_{N-1}}$,
if $kE^{\sqrt{2}}_{AB_i}\geq E^{\sqrt{2}}_{A|B_{i+1}\cdots B_{N-1}}$ for all $i=1,2,\ldots,N-2$,
we have
\begin{equation}\label{EoF8}
\begin{array}{rl}
&E^\beta_{A|B_1B_2\cdots B_{N-1}}\geq E^\beta_{AB_1}+\frac{(1+k)^t-1}{k^t}E^\beta_{AB_2}+\cdots\\
& \ \ \ \ \ \ \ \ \ \ \ \ \ \ \ \ \ \ \ \ \ +\Big(\frac{(1+k)^t-1}{k^t}\Big)^{N-2}E^\beta_{AB_{N-1}},
\end{array}
\end{equation}
for $\beta\geq\sqrt{2}$ and $t=\frac{\beta}{\sqrt{2}}$.
\end{theorem}
As $\frac{(1+k)^t-1}{k^t}\geq2^t-1$ for $t\geq1$ and $0< k\leq 1$,
our new monogamy relations \eqref{EoF3} and \eqref{EoF8} are tighter than the ones given in \cite{SM.Fei1,SM.Fei2,SM.Fei3}.
Also, for $0< k\leq 1$ and $\beta\geq2$, the smaller $k$ is, the tighter inequalities \eqref{EoF3} and \eqref{EoF8} are.

\smallskip
\noindent{\bf Example 2} \, \ Let us again consider the three-qubit state $|\psi\rangle_{ABC}$ defined in \eqref{Con6}
with $\lambda_0=\lambda_3=\frac{1}{2},\ \lambda_2=\frac{\sqrt{2}}{2}$ and $\lambda_1=\lambda_4=0$.
Then $E_{A|BC}=2-\log_{2}3\approx0.811278$,
$E_{AB}=-\frac{2+\sqrt{2}}{4}\log_2\frac{2+\sqrt{2}}{4}-\frac{2-\sqrt{2}}{4}\log_2\frac{2-\sqrt{2}}{4}\approx0.600876$
and $E_{AB}=-\frac{2+\sqrt{3}}{4}\log_2\frac{2+\sqrt{3}}{4}-\frac{2-\sqrt{3}}{4}\log_2\frac{2-\sqrt{3}}{4}\approx0.354579$.
Thus, $E^\beta_{AB}+\big(2^{\frac{\beta}{2}}-1\big)E^\beta_{AC}=(0.600876)^\beta+\big(2^{\frac{\beta}{2}}-1\big)0.354579^\beta$,
$E^\beta_{AB}+\frac{1.5^{\frac{\beta}{2}}-1}{0.5^{\frac{\beta}{2}}}E^\beta_{AC}=
(0.600876)^\beta+\frac{1.5^{\frac{\beta}{2}}-1}{0.5^{\frac{\beta}{2}}}0.354579^\beta$
, $E^\beta_{AB}+\frac{1.7^{\frac{\beta}{2}}-1}{0.7^{\frac{\beta}{2}}}E^\beta_{AC}=
(0.600876)^\beta+\frac{1.7^{\frac{\beta}{2}}-1}{0.7^{\frac{\beta}{2}}}0.354579^\beta$
and $E^\beta_{AB}+\frac{1.9^{\frac{\beta}{2}}-1}{0.9^{\frac{\beta}{2}}}E^\beta_{AC}=
(0.600876)^\beta+\frac{1.9^{\frac{\beta}{2}}-1}{0.9^{\frac{\beta}{2}}}0.354579^\beta$.
One can see that our result is better than the one in \cite{SM.Fei3} for $\beta\geq\sqrt{2}$, hencee better than the ones in \cite{SM.Fei1,SM.Fei2}, see Fig. 2.
\begin{figure}
  \centering
  \includegraphics[width=6cm]{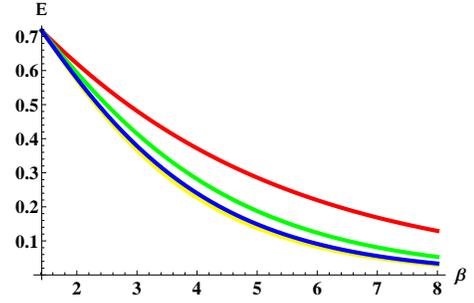}
  \caption{The $y$ axis is the lower bound of the EoF $E^\beta_{A|BC}$.
  The red (green resp. blue) line represents the lower bound from our result for $k=0.5$ ($k=0.7$ resp. $k=0.9$),
  and the yellow line represents the lower bound from the result in \cite{SM.Fei3}.}
\end{figure}

We can also provide tighter polygamy relations for the entanglement of assistance.
The entanglement of assistance (EoA) of $\rho_{AB}$ is defined as \cite{cohen},
\begin{equation}
E^a_{\rho_{AB}}=\max\limits_{\{p_i,|\psi_i\rangle\}}\sum\limits_ip_iE(|\psi_i\rangle),
\end{equation}
with the maximization taking over all possible pure decompositions of $\rho_{AB}$.
For any dimensional multipartite quantum state $\rho_{AB_1B_2\cdots B_{N-1}}$,
a general polygamy inequality of multipartite quantum entanglement was established as \cite{JSK5},
\begin{equation}
E^a(\rho_{A|B_1B_2\cdots B_{N-1}})\leq\sum\limits_{i=1}^{N-1}E^a(\rho_{A|B_i}).
\end{equation}
Using the same approach as for concurrence, we have the following Theorems.
\begin{theorem}
Suppose $k$ is a real number satisfying $0<k\leq 1$.
For any $N$-qubit mixed state $\rho_{AB_1\cdots B_{N-1}}$,
if $kE_{AB_i}^a\geq E_{A|B_{i+1}\cdots B_{N-1}}^a$ for $i=1,2,\ldots,m$, and
$E_{AB_j}^a\leq kE_{A|B_{j+1}\cdots B_{N-1}}^a$ for $j=m+1,\ldots,N-2$, $\forall 1\leq m\leq N-3$, $N\geq4$, we have
\begin{equation}
\begin{array}{rl}
&(E^a_{A|B_1B_2\cdots B_{N-1}})^\beta\\[2.0mm]
&\ \ \leq (E^a_{AB_1})^\beta+\frac{(1+k)^\beta-1}{k^\beta}(E^a_{AB_2})^\beta+\cdots\\[2.0mm]
&\ \ \ \ +\Big(\frac{(1+k)^\beta-1}{k^\beta}\Big)^{m-1}(E^a_{AB_m})^\beta\\[2.0mm]
&\ \ \ \ +\Big(\frac{(1+k)^\beta-1}{k^\beta}\Big)^{m+1}[(E^a_{AB_{m+1}})^\beta+\cdots+(E^a_{AB_{N-2}})^\beta]\\[2.0mm]
&\ \ \ \ +\Big(\frac{(1+k)^\beta-1}{k^\beta}\Big)^m(E^a_{AB_{N-1}})^\beta,
\end{array}
\end{equation}
for $0\leq\beta\leq1$.
\end{theorem}
\begin{theorem}
Suppose $k$ is a real number satisfying $0<k\leq 1$.
For any $N$-qubit mixed state $\rho_{AB_1\cdots B_{N-1}}$,
if $kE_{AB_i}^a\geq E_{A|B_{i+1}\cdots B_{N-1}}^a$ for all $i=1,2,\ldots,N-2$,
we have
\begin{equation}
\begin{array}{rl}
&(E^a_{A|B_1B_2\cdots B_{N-1}})^\beta\leq (E^a_{AB_1})^\beta+\Big(\frac{(1+k)^\beta-1}{k^\beta}\Big)(E^a_{AB_2})^\beta\\[1.5mm]
&\ \ \ \ \ \ \ \ \ \ \ \ \ \ \ \ \ \ \ \ \ \ \ \ + \cdots\\[1.5mm]
& \ \ \ \ \ \ \ \ \ \ \ \ \ \ \ \ \ \ \ \ \ \ \ \ +\Big(\frac{(1+k)^\beta-1}{k^\beta}\Big)^{N-2}(E^a_{AB_{N-1}})^\beta,
\end{array}
\end{equation}
for $0\leq\beta\leq1$.
\end{theorem}

\section{Tighter constraints related to negativity}
The negativity, a well-known quantifier of bipartite entanglement, is defined as
$\mathcal{N}(\rho_{AB})=\big(\|\rho_{AB}^{T_A}\|-1\big)/2$ \cite{Vidal}, where $\rho_{AB}^{T_A}$ is the partial transposed matrix of $\rho_{AB}$ with respect to the subsystem $A$,
and $\|X\|$ denotes the trace norm of $X$, i.e., $\|X\|={\rm tr}\sqrt{XX^{\dag}}$.
For convenient, we use the definition of negativity as $\|\rho_{AB}^{T_A}\|-1$.
Particularly, for any bipartite pure state $|\psi\rangle_{AB}$,
$\mathcal{N}(|\psi\rangle_{AB})=2\sum\limits_{i<j}\sqrt{\lambda_i\lambda_j}=({\rm tr}\sqrt{\rho_A})^2-1$,
where $\lambda_i s$ are the eigenvalues of the reduced density matrix $\rho_A={\rm tr}_B|\psi\rangle_{AB}\langle\psi|$.
The convex-roof extended negativity (CREN) of a mixed state $\rho_{AB}$ is defined by
\begin{equation}
\mathcal{N}_c(\rho_{AB})=\min\limits_{\{p_i,|\psi_i\rangle\}}\sum_{i}p_i\mathcal{N}(|\psi_i\rangle),
\end{equation}
where the minimum is taken over all possible pure state decomposition of $\rho_{AB}$.
Thus $\mathcal{N}_c(\rho_{AB})=\mathcal{C}(\rho_{AB})$ for any two-qubit mixed state $\rho_{AB}$.
The dual to the CREN of a mixed state $\rho_{AB}$ is defined as
\begin{equation}
\mathcal{N}^a_c(\rho_{AB})=\max\limits_{\{p_i,|\psi_i\rangle\}}\sum_{i}p_i\mathcal{N}(|\psi_i\rangle),
\end{equation}
with the maximum taking over all possible pure state decomposition of $\rho_{AB}$.
Furthermore, $\mathcal{N}^a_c(\rho_{AB})=\mathcal{C}^a(\rho_{AB})$ for any two-qubit mixed state $\rho_{AB}$\cite{JSK8}.

Similar to the concurrence and EoF, we have the following Theorems.

\begin{theorem}
Suppose $k$ is a real number satisfying $0<k\leq 1$.
For any $N$-qubit mixed state $\rho_{AB_1\cdots B_{N-1}}$,
if $k\mathcal{N}_{cAB_i}^2\geq\mathcal{N}_{cA|B_{i+1}\cdots B_{N-1}}^2$
for $i=1,2,\ldots,m$, and $\mathcal{N}_{cAB_j}^2\leq k\mathcal{N}_{cA|B_{j+1}\cdots B_{N-1}}^2$ for $j=m+1,\ldots,N-2$,
$\forall 1\leq m\leq N-3$, $N\geq 4$, then we have
\begin{equation}\label{negativity1}
\begin{array}{rl}
&\mathcal{N}^\beta_{cA|B_1\cdots B_{N-1}}\\[2.0mm]
&\ \ \geq\mathcal{N}^\beta_{cAB_1}+\frac{(1+k)^{\frac{\beta}{2}}-1}{k^{\frac{\beta}{2}}}\mathcal{N}^\beta_{cAB_2}+\cdots\\[2.0mm]
&\ \ \ \ +\Big(\frac{(1+k)^{\frac{\beta}{2}}-1}{k^{\frac{\beta}{2}}}\Big)^{m-1}\mathcal{N}^\beta_{cAB_{m}}\\[2.0mm]
&\ \ \ \ +\Big(\frac{(1+k)^{\frac{\beta}{2}}-1}{k^{\frac{\beta}{2}}}\Big)^{m+1}\big(\mathcal{N}^\beta_{cAB_{m+1}}+
\cdots+\mathcal{N}^\beta_{cAB_{N-2}}\big)\\[2.0mm]
&\ \ \ \ +\Big(\frac{(1+k)^{\frac{\beta}{2}}-1}{k^{\frac{\beta}{2}}}\Big)^m\mathcal{N}^\beta_{cAB_{N-1}}
\end{array}
\end{equation}
for all $\beta\geq2$.
\end{theorem}
\begin{theorem}
Suppose $k$ is a real number satisfying $0<k\leq 1$.
For any $N$-qubit mixed state $\rho_{AB_1\cdots B_{N-1}}$,
if all $k\mathcal{N}_{cAB_i}^2\geq\mathcal{N}_{cA|B_{i+1}\cdots B_{N-1}}^2$ for all $i=1,2,\ldots,N-2$, then
\begin{equation}\label{negativity2}
\begin{array}{rcl}
\mathcal{N}^\beta_{cA|B_1\cdots B_{N-1}}
&\geq&\mathcal{N}^\beta_{cAB_1}+\frac{(1+k)^{\frac{\beta}{2}}-1}
{k^{\frac{\beta}{2}}}\mathcal{N}^\beta_{cAB_2}+\cdots\\[3.0mm]
&&+\Big(\frac{(1+k)^{\frac{\beta}{2}}-1}{k^{\frac{\beta}{2}}}\Big)^{N-2}\mathcal{N}^\beta_{cAB_{N-1}}
\end{array}
\end{equation}
for $\beta\geq2$.
\end{theorem}

\smallskip
\noindent{\bf Example 3} \, \ Consider the state in Example 1 with
$\lambda_0=\lambda_3=\frac{1}{2},\ \lambda_2=\frac{\sqrt{2}}{2}$ and $\lambda_1=\lambda_4=0$.
We have $\mathcal{N}_{cA|BC}=\frac{\sqrt{3}}{2}$, $\mathcal{N}_{cAB}=\frac{\sqrt{2}}{2}$ and $\mathcal{C}_{cAC}=\frac{1}{2}$.
Then $\mathcal{N}_{cAB}^\beta+\big(2^{\frac{\beta}{2}}-1\big)\mathcal{N}_{cAC}^\beta=\big(\frac{\sqrt{2}}{2}\big)^\beta+
\big(2^{\frac{\beta}{2}}-1\big)\big(\frac{1}{2}\big)^\beta$ and
$\mathcal{N}_{cAB}^\beta+\frac{(1+k)^\frac{\beta}{2}-1}{k^{\frac{\beta}{2}}}\mathcal{N}_{cAC}^\beta=\big(\frac{\sqrt{2}}{2}\big)^\beta+
\frac{(1+k)^\frac{\beta}{2}-1}{k^{\frac{\beta}{2}}}\big(\frac{1}{2}\big)^\beta$.
One can see that our result is better than the one in \cite{SM.Fei3} for $\beta\geq2$, thus also better than the ones in \cite{SM.Fei1,SM.Fei2}, see Fig. 3.
\begin{figure}
  \centering
  \includegraphics[width=6cm]{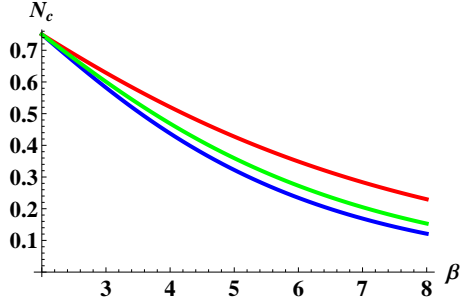}
  \caption{The $y$ axis is the lower bound of the negativity $\mathcal{N}_c(|\psi\rangle_{A|BC})$, which are functions of $\beta$.
  The red (green) line represents the lower bound from our result for $k=0.6$ ($k=0.8$),
  and the blue line represents the lower bound from the result in \cite{SM.Fei3}.}
\end{figure}

For the negativity of assistance $\mathcal{N}_c^a$, we have the following results.
\begin{theorem}
Suppose $k$ is a real number satisfying $0<k\leq 1$.
For an $N$-qubit pure state $|\psi\rangle_{A|B_1\cdots B_{N-1}}$,
if $k(\mathcal{N}^a_{cA|B_i})^2\geq(\mathcal{N}^a_{cA|B_{i+1}\cdots B_{N-1}})^2$
for $i=1,2,\ldots,m$, and $(\mathcal{N}^a_{cAB_j})^2\leq k(\mathcal{N}^a_{cA|B_{j+1}\cdots B_{N-1}})^2$ for $j=m+1,\ldots,N-2$,
$\forall 1\leq m\leq N-3$, $N\geq 4$, then we have
\begin{equation}
\begin{array}{rl}
&[\mathcal{N}_c^a(|\psi\rangle_{A|B_1\cdots B_{N-1}})]^{\beta}\\[1.5mm]
&\ \ \leq(\mathcal{N}^a_{cAB_1})^{\beta}+\Big(\frac{(1+k)^{\frac{\beta}{2}}-1}{k^{\frac{\beta}{2}}}\Big)(\mathcal{N}^a_{cAB_2})^2+\cdots\\[1.5mm]
&\ \ \ \ +\Big(\frac{(1+k)^{\frac{\beta}{2}}-1}{k^{\frac{\beta}{2}}}\Big)^{m-1}(\mathcal{N}^a_{cAB_{m}})^\beta\\[1.5mm]
&\ \ \ \ +\Big(\frac{(1+k)^{\frac{\beta}{2}}-1}{k^{\frac{\beta}{2}}}\Big)^{m+1}[(\mathcal{N}^a_{cAB_{m+1}})^\beta+\cdots\\[1.5mm]
&\ \ \ \ +(\mathcal{N}^a_{cAB_{N-2}})^\beta]\\[1.5mm]
&\ \ \ \ +\Big(\frac{(1+k)^{\frac{\beta}{2}}-1}{k^{\frac{\beta}{2}}}\Big)^m(\mathcal{N}^a_{c AB_{N-1}})^\beta
\end{array}
\end{equation}
for all $0\leq\beta\leq2$.
\end{theorem}
\begin{theorem}
Suppose $k$ is a real number satisfying $0<k\leq 1$.
For any $N$-qubit mixed state $|\psi\rangle_{AB_1\cdots B_{N-1}}$,
if $k(\mathcal{N}^a_{cAB_i})^2\geq(\mathcal{N}^a_{cA|B_{i+1}\cdots B_{N-1}})^2$
for all $i=1,2,\ldots,N-2$, then
\begin{equation}
\begin{array}{rl}
&[\mathcal{N}^a_c(|\psi\rangle_{A|B_1\cdots B_{N-1}})]^\beta\\[2.0mm]
&\ \ \leq(\mathcal{N}^a_{cAB_1})^\beta+\Big(\frac{(1+k)^{\frac{\beta}{2}}-1}{k^{\frac{\beta}{2}}}\Big)(\mathcal{N}^a_{cAB_2})^\beta+\cdots\\[2.0mm]
&\ \ \ \ +\Big(\frac{(1+k)^{\frac{\beta}{2}}-1}{k^{\frac{\beta}{2}}}\Big)^{N-2}(\mathcal{N}^a_{cAB_{N-1}})^\beta
\end{array}
\end{equation}
for $0\leq\beta\leq2$.
\end{theorem}

\section{Tighter monogamy relations for Tsallis-$q$ entanglement and R\'enyi-$\alpha$ entanglement}
In this section, we study the Tsallis-$q$ entanglement and R\'enyi-$\alpha$ entanglement,
and establish the corresponding monogamy and polygamy relations for the two entanglement measures, respectively.

\subsection{Tighter monogamy and polygamy relations for Tsallis-$q$ entanglement}

The Tsallis-$q$ entanglement of a bipartite pure state $|\psi\rangle_{AB}$ is defined as \cite{JSK3}
\begin{equation}
T_q(|\psi\rangle_{AB})=S_q(\rho_A)=\frac{1}{q-1}(1-{\rm tr}\rho_A^q),
\end{equation}
where $q>0$ and $q\neq1$.
For the case $q$ tends to 1, $T_q(\rho)$ is just the von Neumann entropy, $\lim\limits_{q\rightarrow 1}T_q(\rho)=-{\rm tr}\rho\log_2\rho=S(\rho)$.
The Tsallis-$q$ entanglement of a bipartite mixed state $\rho_{AB}$ is given by
$T_q(\rho_{AB})=\min\limits_{\{p_i,|\psi_i\rangle\}}\sum\limits_{i}p_iT_q(|\psi_i\rangle)$
with the minimum taken over all possible pure state decompositions of $\rho_{AB}$.
For $\frac{5-\sqrt{13}}{2}\leq q\leq\frac{5+\sqrt{13}}{2}$, Yuan {\it et al.} proposed an analytic relationship
between the Tsallis-$q$ entanglement and concurrence,
\begin{equation}\label{Tq1}
T_q(|\psi\rangle_{AB})=g_q(\mathcal{C}^2(|\psi\rangle_{AB})),
\end{equation}
where
\begin{equation}\label{Tq2}
g_q(x)=\frac{1}{q-1}\Big[1-\Big(\frac{1+\sqrt{1-x}}{2}\Big)^q-\Big(\frac{1-\sqrt{1-x}}{2}\Big)^q\Big]
\end{equation}
with $0\leq x\leq1$ \cite{Yuan}.
It has also been proved that $T_q(|\psi\rangle)=g_q(\mathcal{C}^2(|\psi\rangle))$ if $|\psi\rangle$ is a $2\otimes m$ pure state,
and $T_q(\rho)=g_q(\mathcal{C}^2(\rho))$ if $\rho$ is a two-qubit mixed state.
Hence, \eqref{Tq1} holds for any $q$ such that $g_q(x)$ in \eqref{Tq2} is monotonically increasing and convex.
Particularly, one has that
\begin{equation}
g_q(x^2+y^2)\geq g_q(x^2)+g_q(y^2)
\end{equation}
for $2\leq q\leq 3$.
In Ref. \cite{JSK3}, Kim provided a monogamy relation for the Tsallis-$q$ entanglement,
\begin{equation}\label{Tq4}
T_{qA|B_1B_2\cdots B_{N-1}}\geq\sum\limits_{i=1}^{N-1}T_{qA|B_i},
\end{equation}
where $i=1,2,\ldots,N-1$ and $2\leq q\leq 3$.
Later, this relation was improved as follows:
if $\mathcal{C}_{AB_i}\geq\mathcal{C}_{A|B_{i+1}\cdots B_{N-1}}$
for $i=1,2,\ldots,m$, and $\mathcal{C}_{AB_j}\leq\mathcal{C}_{A|B_{j+1}\cdots B_{N-1}}$
for $j=m+1,\ldots,N-2$, $\forall 1\leq m\leq N-3$, $N\geq4$, then
\begin{equation}\label{Tq3}
\begin{array}{rl}
&T^\beta_{qA|B_1B_2\cdots B_{N-1}}\\
&\ \ \geq T^\beta_{qA|B_1}+(2^\beta-1)T^\beta_{qA|B_2}+\cdots+(2^\beta-1)^{m-1}T^\beta_{qA|B_m}\\
&\ \ \ \ +(2^\beta-1)^{m+1}(T^\beta_{qA|B_{m+1}}+\cdots+T^\beta_{qA|B_{N-2}})\\
&\ \ \ \ +(2^\beta-1)^mT^\beta_{qA|B_{N-1}},
\end{array}
\end{equation}
where $\beta\geq1$ and $T^\beta_{qA|B_1B_2\cdots B_{N-1}}$ quantifies the Tsallis-$q$ entanglement under partition $A|B_1B_2\cdots B_{N-1}$,
and $T^\beta_{qA|B_i}$ quantifies that of the two-qubit subsystem $AB_i$ with $2\leq q\leq3$.
Moreover, for $\frac{5-\sqrt{13}}{2}\leq q\leq\frac{5+\sqrt{13}}{2}$, one has
\begin{equation}
T^2_{qA|B_1B_2\cdots B_{N-1}}\geq\sum\limits_{i=1}^{N-1}T^2_{qA|B_i}.
\end{equation}

We now provide monogamy relations which are tighter than \eqref{Tq4} and \eqref{Tq3}.
\begin{theorem}
Suppose $k$ is a real number satisfying $0<k\leq 1$.
For an arbitrary $N$-qubit mixed state $\rho_{AB_1B_2\cdots B_{N-1}}$, if
$kT_{qAB_i}\geq T_{qA|B_{i+1}\cdots B_{N-1}}$ for $i=1,2,\ldots,m$, and
$T_{qAB_j}\leq kT_{qA|B_{j+1}\cdots B_{N-1}}$ for $j=m+1,\ldots,N-2$, $\forall 1\leq m\leq N-3$, $N\geq4$,
then we have
\begin{equation}\label{Tq5}
\begin{array}{rl}
&T^\beta_{qA|B_1B_2\cdots B_{N-1}}\\[1.5mm]
&\ \ \geq T^\beta_{qA|B_1}+\frac{(1+k)^\beta-1}{k^\beta}T^\beta_{qA|B_2}+\cdots\\[2.0mm]
&\ \ \ \ +\Big(\frac{(1+k)^\beta-1}{k^\beta}\Big)^{m-1}T^\beta_{qA|B_m}\\[2.0mm]
&\ \ \ \ +\Big(\frac{(1+k)^\beta-1}{k^\beta}\Big)^{m+1}(T^\beta_{qA|B_{m+1}}+\cdots+T^\beta_{qA|B_{N-2}})\\[2.0mm]
&\ \ \ \ +\Big(\frac{(1+k)^\beta-1}{k^\beta}\Big)^mT^\beta_{qA|B_{N-1}},
\end{array}
\end{equation}
for $\beta\geq1$ and $2\leq q\leq3$.
\end{theorem}
\begin{theorem}
Suppose $k$ is a real number satisfying $0<k\leq 1$.
For any $N$-qubit mixed state $\rho_{AB_1\cdots B_{N-1}}$,
if all $kT_{qAB_i}\geq T_{qA|B_{i+1}\cdots B_{N-1}}$ for $i=1,2,\ldots,N-2$, then we have
\begin{equation}\label{Tq6}
\begin{array}{rcl}
T_{qA|B_1\cdots B_{N-1}}^\beta
&\geq& T_{qAB_1}^\beta+\frac{(1+k)^{\beta}-1}{k^{\beta}}T_{qAB_2}^\beta+\cdots\\[3.0mm]
&&+\Big(\frac{(1+k)^{\beta}-1}{k^{\beta}}\Big)^{N-2}T_{qAB_{N-1}}^\beta,
\end{array}
\end{equation}
for $\beta\geq1$ and $2\leq q\leq3$.
\end{theorem}

\smallskip
\noindent{\bf Example 4} \, \ Consider the quantum state given in Example 1 with
$\lambda_0=\lambda_3=\frac{1}{2},\ \lambda_2=\frac{\sqrt{2}}{2}$ and $\lambda_1=\lambda_4=0$.
For $q=2$, one has $T_{2A|BC}=\frac{3}{8}$, $T_{2AB}=\frac{1}{4}$ and $T_{2AC}=\frac{1}{8}$.
Then $T_{2AB}^\beta+\big(2^{\beta}-1\big)T_{2AC}^\beta=(\frac{1}{4})^\beta+\big(2^{\beta}-1\big)(\frac{1}{8})^\beta$ and
$T_{2AB}^{\beta}+\frac{(1+k)^\beta-1}{k^{\beta}}T_{2AC}^{\beta}=\big(\frac{1}{4}\big)^{\beta}+
\frac{(1+k)^\beta-1}{k^{\beta}}\big(\frac{1}{8}\big)^\beta$.
It can be seen that our result is better than the one in \cite{SM.Fei3} for $\beta\geq1$, and also better than the ones given in \cite{SM.Fei1,SM.Fei2}, see Fig. 4.
\begin{figure}
  \centering
  \includegraphics[width=6cm]{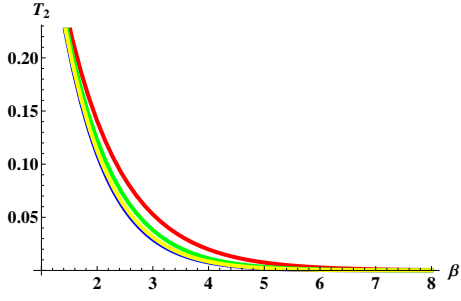}
  \caption{The $y$ axis is the lower bound of the Tsallis-$q$ entanglement $T_q^\beta(|\psi\rangle_{A|BC})$.
  The red (green resp. yellow) line represents the lower bound from our result for $k=0.5$ ($k=0.7$ resp. $k=0.9$),
  and the blue line represents the lower bound from the result in \cite{SM.Fei3}.}
\end{figure}

As a dual quantity to Tsallis-$q$ entanglement, the Tsallis-$q$ entanglement of assistance (TEoA) is defined by \cite{JSK3},
$T_q^a(\rho_{AB})=\max\limits_{\{p_i,|\psi_i\rangle\}}\sum\limits_ip_iT_q(|\psi_i\rangle)$,
where the maximum is taken over all possible pure state decompositions of $\rho_{AB}$.
If $1\leq q\leq 2$ or $3\leq q\leq4$, the function $g_q$ defined in \eqref{Tq2} satisfies
\begin{equation}
g_q(\sqrt{x^2+y^2})\leq g_q(x)+g_q(y),
\end{equation}
which leads to the Tsallis polygamy inequality
\begin{equation}\label{Tq7}
T^a_{qA|B_1B_2\cdots B_{N-1}}\leq\sum\limits_{i=1}^{N-1}T^a_{qA|B_i}
\end{equation}
for any multi-qubit state $\rho_{A|B_1B_2\cdots B_{N-1}}$ \cite{JSK9}.
Here we provide tighter polygamy relations related to Tsallis-$q$ entanglement.
We have the following results.
\begin{theorem}
Suppose $k$ is a real number satisfying $0<k\leq 1$.
For any $N$-qubit mixed state $\rho_{AB_1\cdots B_{N-1}}$,
if $kT^a_{qAB_i}\geq T^a_{qA|B_{i+1}\cdots B_{N-1}}$ for $i=1,2,\ldots,m$, and
$T^a_{qAB_j}\leq kT^a_{qA|B_{j+1}\cdots B_{N-1}}$ for $j=m+1,\ldots,N-2$, $\forall 1\leq m\leq N-3$, $N\geq4$, then
\begin{equation}\label{Tq8}
\begin{array}{rl}
&(T^a_{qA|B_1B_2\cdots B_{N-1}})^\beta\\[2.0mm]
&\ \ \leq (T^a_{qAB_1})^\beta+\frac{(1+k)^\beta-1}{k^\beta}(T^a_{qAB_2})^\beta+\cdots\\[2.0mm]
&\ \ \ \ +\Big(\frac{(1+k)^\beta-1}{k^\beta}\Big)^{m-1}(T^a_{qAB_m})^\beta\\[2.0mm]
&\ \ \ \ +\Big(\frac{(1+k)^\beta-1}{k^\beta}\Big)^{m+1}[(T^a_{qAB_{m+1}})^\beta+\cdots+(T^a_{qAB_{N-2}})^\beta]\\[2.0mm]
&\ \ \ \ +\Big(\frac{(1+k)^\beta-1}{k^\beta}\Big)^m(T^a_{qAB_{N-1}})^\beta,
\end{array}
\end{equation}
for $0\leq\beta\leq1$ with $1\leq q\leq2$ or $3\leq q\leq4$.
\end{theorem}
\begin{theorem}
Suppose $k$ is a real number satisfying $0<k\leq 1$.
For any $N$-qubit mixed state $\rho_{AB_1\cdots B_{N-1}}$,
if $kT^a_{qAB_i}\geq T^a_{qA|B_{i+1}\cdots B_{N-1}}$ for all $i=1,2,\ldots,N-2$,
we have
\begin{equation}\label{Tq9}
\begin{array}{rl}
&T^\beta_{qA|B_1B_2\cdots B_{N-1}}\leq T^\beta_{qAB_1}+\Big(\frac{(1+k)^\beta-1}{k^\beta}\Big)T^\beta_{qAB_2}+\cdots\\[2.0mm]
& \ \ \ \ \ \ \ \ \ \ \ \ \ \ \ \ \ \ \ \ \ +\Big(\frac{(1+k)^\beta-1}{k^\beta}\Big)^{N-2}T^\beta_{qAB_{N-1}},
\end{array}
\end{equation}
for $0\leq\beta\leq1$ with $1\leq q\leq2$ or $3\leq q\leq4$.
\end{theorem}

\subsection{Tighter monogamy and polygamy relations for R\'enyi-$\alpha$ entanglement}
For a bipartite pure state $|\psi\rangle_{AB}$, the R\'enyi-$\alpha$ entanglement is defined as \cite{vidal} $E(|\psi\rangle_{AB})=S_\alpha(\rho_A)$,
where $S_\alpha(\rho)=\frac{1}{1-\alpha}\log_2{\rm tr}\rho^\alpha$ for any $\alpha>0$
and $\alpha\neq1$, and $\lim\limits_{\alpha\rightarrow 1}S_\alpha(\rho)=S(\rho)=-{\rm tr}\rho\log_2\rho$.
For a bipartite mixed state $\rho_{AB}$, the R\'enyi-$\alpha$ entanglement is given by
\begin{equation}
E_\alpha(\rho_{AB})=\min\limits_{\{p_i,|\psi_i\rangle\}}\sum\limits_{i}p_iE_\alpha(|\psi_i\rangle),
\end{equation}
where the minimum is taken over all possible pure-state decompositions of $\rho_{AB}$.
For each $\alpha>0$, one has $E_\alpha(\rho_{AB})=f_\alpha(\mathcal{C}(\rho_{AB}))$,
where $f_\alpha(x)=\frac{1}{1-\alpha}\log\big[\big(\frac{1-\sqrt{1-x^2}}{2}\big)^2+(\frac{1+\sqrt{1-x^2}}{2}\big)^2\big]$
is a monotonically increasing and convex function \cite{JSK2}.
For $\alpha\geq2$ and any $n$-qubit state $\rho_{A|B_1B_2\cdots B_{N-1}}$, one has \cite{JSK3}
\begin{equation}\label{entropy1}
\begin{array}{rl}
&E_{\alpha A|B_1B_2\cdots B_{N-1}}\\[2mm]
&\ \ \geq E_{\alpha A|B_1}+E_{\alpha A|B_2}+\cdots+E_{\alpha A|B_{N-1}}.
\end{array}
\end{equation}

We propose the following two monogamy relations for the R\'enyi-$\alpha$ entanglement, which are tighter than the previous results.
\begin{theorem}
Suppose $k$ is a real number satisfying $0<k\leq 1$.
For an arbitrary $N$-qubit mixed state $\rho_{AB_1B_2\cdots B_{N-1}}$, if
$kE_{\alpha AB_i}\geq E_{\alpha A|B_{i+1}\cdots B_{N-1}}$ for $i=1,2,\ldots,m$, and
$E_{\alpha AB_j}\leq kT_{\alpha A|B_{j+1}\cdots B_{N-1}}$ for $j=m+1,\ldots,N-2$, $\forall 1\leq m\leq N-3$, $N\geq4$,
then
\begin{equation}
\begin{array}{rl}
&(E_{\alpha A|B_1B_2\cdots B_{N-1}})^{\beta}\\[2.0mm]
&\ \ \geq (E_{\alpha A|B_1})^{\beta}+\frac{(1+k)^\beta-1}{k^\beta}(E_{\alpha A|B_2})^{\beta}+\cdots\\[2.0mm]
&\ \ \ \ +\Big(\frac{(1+k)^\beta-1}{k^\beta}\Big)^{m-1}(E_{\alpha A|B_m})^{\beta}\\[2.0mm]
&\ \ \ \ +\Big(\frac{(1+k)^\beta-1}{k^\beta}\Big)^{m+1}[(E_{\alpha A|B_{m+1}})^{\beta}+\cdots+(E_{\alpha A|B_{N-2}})^{\beta}]\\[2.0mm]
&\ \ \ \ +\Big(\frac{(1+k)^\beta-1}{k^\beta}\Big)^m(E_{\alpha A|B_{N-1}})^{\beta},
\end{array}
\end{equation}
for $\beta\geq1$ and $\alpha\geq2$.
\end{theorem}
\begin{theorem}
Suppose $k$ is a real number satisfying $0<k\leq 1$.
For an arbitrary $N$-qubit mixed state $\rho_{AB_1B_2\cdots B_{N-1}}$,
if $kE_{\alpha AB_i}\geq E_{\alpha A|B_{i+1}\cdots B_{N-1}}$ for all $i=1,2,\ldots,N-2$, then
\begin{equation}
\begin{array}{rl}
&(E_{\alpha A|B_1\cdots B_{N-1}})^{\beta}\\[2.0mm]
&\ \ \geq (E_{\alpha AB_1})^{\beta}+\Big(\frac{(1+k)^{\alpha}-1}{k^{\alpha}}\Big)(E_{\alpha AB_2})^{\beta}+\cdots\\[2.0mm]
&\ \ \ \ +\Big(\frac{(1+k)^{\alpha}-1}{k^{\alpha}}\Big)^{N-2}(E_{\alpha AB_{N-1}})^{\beta}
\end{array}
\end{equation}
for $\beta\geq1$ and $\alpha\geq2$.
\end{theorem}

\smallskip
\noindent{\bf Example 5} \, \ Consider again the state given in Example 1 with
$\lambda_0=\lambda_3=\frac{1}{2},\ \lambda_2=\frac{\sqrt{2}}{2}$ and $\lambda_1=\lambda_4=0$.
For $\alpha=2$, we find $E_{2A|BC}=\log_2\frac{8}{5}\approx678072$, $E_{2AB}=\log_2\frac{8}{7}\approx0.415037$
and $E_{2AC}=\log_2\frac{4}{3}\approx0.192645$.
Then $E_{2AB}^\alpha+E_{2AC}^\alpha=0.415037^\alpha+0.192645^\alpha$ and
$E_{2AB}^{\alpha}+\frac{(1+k)^\alpha-1}{k^{\alpha}}E_{2AC}^{\alpha}=0.415037^{\alpha}+\frac{(1+k)^\alpha-1}{k^{\alpha}}0.192645^\alpha$.
One can see that our result is better than the result in \cite{JSK3}, and the smaller $k$ is, the tighter relation is, see Fig. 5.
\begin{figure}
  \centering
  \includegraphics[width=6cm]{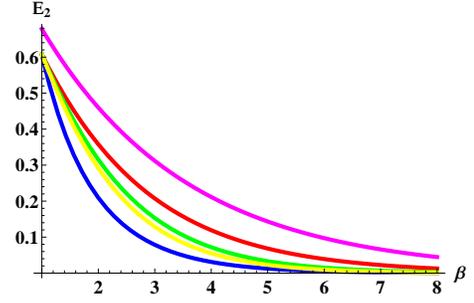}
  \caption{The $y$ axis is the lower bound of the R\'enyi entropy entanglement  $E^\beta_{2}(|\psi\rangle_{A|BC})$.
  The purple line represents the value of $E^\beta_2(|\psi\rangle_{A|BC})$,
  the red (green resp. yellow) line represents the lower bound from our result for $k=0.5$ ($k=0.7$ resp. $k=0.9$),
  and the blue line represents the lower bound from the result (\ref{entropy1}) in \cite{JSK3}.}
\end{figure}

The R\'enyi-$\alpha$ entanglement of assistance (REoA),
a dual quantity to R\'enyi-$\alpha$ entanglement, is defined as
$E_\alpha^a(\rho_{AB})=\max\limits_{\{p_i,|\psi_i\rangle\}}\sum\limits_{i}p_iE_\alpha(|\psi_i\rangle)$,
where the maximum is taken over all possible pure state decompositions of $\rho_{AB}$.
For $\alpha\in[\frac{\sqrt{7}-1}{2},\frac{\sqrt{13}-1}{2}]$ and any $n$-qubit state $\rho_{AB_1B_2\cdots B_{N-1}}$,
a polygamy relation of multi-partite quantum entanglement in terms of REoA has been given by \cite{Song}:
\begin{equation}
\begin{array}{rl}
&E^a_{\alpha A|B_1B_2\cdots B_{N-1}}\\[2mm]
&\ \ \leq E^a_{\alpha A|B_1}+E^a_{\alpha A|B_2}+\cdots +E^a_{\alpha A|B_{N-1}}.
\end{array}
\end{equation}
We improve this inequality to be a tighter ones under some netural conditions.
\begin{theorem}
Suppose $k$ is a real number satisfying $0<k\leq 1$.
For an arbitrary $N$-qubit mixed state $\rho_{AB_1B_2\cdots B_{N-1}}$, if
$kE^a_{\alpha AB_i}\geq E^a_{\alpha A|B_{i+1}\cdots B_{N-1}}$ for $i=1,2,\ldots,m$, and
$E^a_{\alpha AB_j}\leq kE^a_{\alpha A|B_{j+1}\cdots B_{N-1}}$ for $j=m+1,\ldots,N-2$, $\forall 1\leq m\leq N-3$, $N\geq4$,
then
\begin{equation}
\begin{array}{rl}
&(E^a_{\alpha A|B_1B_2\cdots B_{N-1}})^{\beta}\\[2.0mm]
&\ \ \leq (E^a_{\alpha A|B_1})^{\beta}+\frac{(1+k)^\beta-1}{k^\beta}(E^a_{\alpha A|B_2})^{\beta}+\cdots\\[2.0mm]
&\ \ \ \ +\Big(\frac{(1+k)^\beta-1}{k^\beta}\Big)^{m-1}(E^a_{\alpha A|B_m})^{\beta}\\[2.0mm]
&\ \ \ \ +\Big(\frac{(1+k)^\beta-1}{k^\beta}\Big)^{m+1}\Big[(E^a_{\alpha A|B_{m+1}})^{\beta}+\cdots+(E^a_{\alpha A|B_{N-2}})^{\beta}\Big]\\[2.0mm]
&\ \ \ \ +\Big(\frac{(1+k)^\beta-1}{k^\beta}\Big)^m(E^a_{\alpha A|B_{N-1}})^{\beta},
\end{array}
\end{equation}
for $0\leq\beta\leq1$ with $\frac{\sqrt{7}-1}{2}\leq\alpha\leq\frac{\sqrt{13}-1}{2}$.
\end{theorem}
\begin{theorem}
Suppose $k$ is a real number satisfying $0<k\leq 1$.
For an arbitrary $N$-qubit mixed state $\rho_{AB_1B_2\cdots B_{N-1}}$,
if $kE^a_{\alpha AB_i}\geq E^a_{\alpha A|B_{i+1}\cdots B_{N-1}}$ for all $i=1,2,\ldots,N-2$, then
\begin{equation}
\begin{array}{rl}
&(E^a_{\alpha  A|B_1\cdots B_{N-1}})^{\beta}\\[2.0mm]
&\ \ \leq (E^a_{\alpha AB_1})^{\beta}+\Big(\frac{(1+k)^{\alpha}-1}{k^{\alpha}}\Big)(E^a_{\alpha AB_2})^{\beta}+\cdots\\[2.0mm]
&\ \ \ \ +\Big(\frac{(1+k)^{\alpha}-1}{k^{\alpha}}\Big)^{N-2}(E^a_{\alpha AB_{N-1}})^{\beta}
\end{array}
\end{equation}
for $0\leq\beta\leq1$, with $\frac{\sqrt{7}-1}{2}\leq\alpha\leq\frac{\sqrt{13}-1}{2}$.
\end{theorem}

\section{Conclusion}
Both entanglement monogamy and polygamy are fundamental properties of multipartite entangled states.
We have presented monogamy relations related to the $\beta$th power of concurrence, entanglement of formation,
negativity, Tsallis-$q$ and R\'enyi-$\alpha$ entanglement.
We also provide polygamy relations related to these entanglement measures.
All the relations we presented in this paper are tighter than the previous results.
These tighter monogamy and polygamy inequalities can also provide finer characterizations of the entanglement distributions
among the multiqubit systems.
Our results provide a rich reference for future work on the study of multiparty quantum entanglement.
And our approaches are also useful for further study on the monogamy and polygamy properties related to measures of other quantum correlations and quantum coherence \cite{framework,coherence}.

\section{Acknowledgements}
This work is supported by the National Natural Science Foundation of China under Grant Nos. 11805143 and 11675113, and Key Project of Beijing Municipal Commission of Education under No. KZ201810028042.

\end{document}